\documentclass[12pt]{article}
\usepackage{amsfonts}
\usepackage{amssymb}
\usepackage{color}
\usepackage{graphicx}

\newcommand{\Ao}{{\cal A}_0}

\newcommand{\Sc}{{\cal S}}

\newcommand{\mc}{\mathcal}
\newcommand{\norm}[1]{ \parallel #1 \parallel}
\newcommand{\be}{\begin{equation}}
	\newcommand{\vp}{\varphi}
\newcommand{\en}{\end{equation}}
\newcommand{\bea}{\begin{eqnarray}}
\newcommand{\ena}{\end{eqnarray}}
\newcommand{\beano}{\begin{eqnarray*}}
\newcommand{\enano}{\end{eqnarray*}}
\newcommand{\D}{{\mc D}}

\newcommand{\G}{{\cal G}}
\newcommand{\A}{{\cal A}}

\newcommand{\M}{{\cal M}}

\newcommand{\ltwo}{{\Lc^2(\mathbb{R})}}

\newcommand{\N}{{\mathfrak N}}

\newcommand{\Lc}{{\cal L}}
\newcommand{\1}{1 \!\! 1}

\newcommand{\Hil}{\mc H}

\newtheorem{thm}{Theorem}

\newtheorem{lemma}[thm]{Lemma}
\newtheorem{prop}[thm]{Proposition}
\newtheorem{defn}[thm]{Definition}

\newtheorem{exe}[thm]{Example}
\newenvironment{proof}{\noindent {\bf Proof:}}{\hfill$\Box$}

\catcode `\@=11 \@addtoreset{equation}{section}
\def\theequation{\arabic{section}.\arabic{equation}}
\catcode `\@=12
\begin{document}

\begin{center}
{\Large \textbf{Sesquilinear forms as eigenvectors in quasi *-algebras, with an application to ladder elements}} \vspace{1.2cm%
}\\[0pt]

{\large F. Bagarello}
\vspace{2mm}\\[0pt]
Dipartimento di Ingegneria,\\[0pt]
Universit\`{a} di Palermo, I - 90128 Palermo,\\
and I.N.F.N., Sezione di Catania\\
E-mail: fabio.bagarello@unipa.it\\

\vspace{0.3cm}

{\large H. Inoue}
\vspace{2mm}\\[0pt]
Department of Economics,\\[0pt]
 Kyushu Sangyo University, Fukuoka, Japan,\\
E-mail:  h-inoue@ip.kyusan-u.ac.jp\\

\vspace{3mm}

{\large S. Triolo}
\vspace{2mm}\\[0pt]
Dipartimento di Ingegneria,\\[0pt]
Universit\`{a} di Palermo, I - 90128 Palermo,\\
E-mail: salvatore.triolo@unipa.it\\

\vspace{1mm}

\end{center}

\vspace*{1.1cm}

\begin{abstract}
\noindent We consider a particular class of sesquilinear forms on a  {Banach quasi *-algebra} $(\A[\|.\|],\Ao[\|.\|_0])$ which we call {\em eigenstates of an element} $a\in\A$, and we deduce some of their properties. We further apply our definition to a family of ladder elements, i.e. elements of $\A$ obeying certain commutation relations physically motivated, and we discuss several results, including orthogonality and biorthogonality of the forms, via GNS-representation.

\end{abstract}

\vspace*{.4cm}


\vspace{2mm}
{\bf Corresponding author:--} Fabio Bagarello

\vfill

\newpage


\section{Introduction and preliminaries}\label{sect1}

The role of eigenstates it is quite well known both in pure and in applied mathematics, and it is particularly relevant in quantum mechanics. In this case, for instance, eigenvectors of a given Hamiltonian (which usually describes the energy of a given physical system $\G$) are typically interpreted as the stationary states of $\G$. This means that, if $\G$ is {\em prepared} in one specific eigenstate at the initial time $t=0$, then, in absence of other effects, $\G$ {\em stays in that state}. This is one of the (physical) reasons why the properties of the set of all the eigenstates of a given operator have been studied at length during the past decades. A more mathematical reason for this interest is that, in many situations, this set spans the whole Hilbert space where $\G$ is defined, \cite{hall,mess}. Quite generally, in the literature, eigenstates are meant to be vectors in the Hilbert space. Recently, \cite{denittis}, the authors considered the notion of eigenvectors in a C*-algebraic settings, replacing vectors with linear functional. C*-algebras are often used in connection with quantum systems with infinite degrees of freedom in view of the possibility of using the same underlying structure to describe different physical situations, using inequivalent representations of the abstract C*-algebra describing the system,  \cite{br}-\cite{sewbook2} . 

In this paper we further extend the notion of eigenvectors to  Banach quasi *-algebra, by using sesquilinear forms, \cite{mariacamillo}. This extension can be useful in presence of unbounded operators, which often appears when dealing with concrete physical systems, \cite{aitbook,bag1,schu}, and which are not so naturally analyzed in C*-algebras. 

The paper is organized as follows: after some preliminaries, we introduce our definition of eigenstates for sesquilinear forms, and we deduce several properties. In particular we show that from any such form it is possible to define a second form which is again an eigenstate, but of a different operator. Section \ref{sect3} is devoted to a detailed example based on a class of ladder operators known as quons, \cite{moh,fiv,gree}. In particular we show how a family of eigenstates of a quonic number-like operator $n_0$ can be defined, and we also discuss the {\em ortogonality} of these states, after explaining what orthogonality is for us, in our context. We both consider the cases of $n_0=n_0^*$, and what happens if this equality does not hold. In this latter case we show that a sort of biorthogonality between forms can be introduced, if some suitable conditions are satisfied. Section \ref{sect4} contains our conclusions. To keep the paper self-contained, in the Appendix we include some basic facts on quons, relevant for what we will do in Section \ref{sect3}.

\subsection{Basic notions}
\vspace{1ex} We briefly recall here some definitions and facts
needed in the sequel.
 Let $\A$ be a complex vector space and $\Ao$ a  $^\ast$ -algebra contained in $\A$. $\A$ is said
a \textit{quasi  $^\ast$-algebra with distinguished  $^\ast$-algebra $\Ao$} (or, simply, over $\Ao$) if
\begin{itemize}\item[(i)]\label{11} the left multiplication $ax$ and the right multiplication $
	xa$ of an element $a$ of $\A$ and an element $x$ of $\A_0$ which
	extend the multiplication of $\A_0$ are always defined and
	bilinear; \item[(ii)] $x_1 (x_2 a)= (x_1x_2 )a$ and $x_1(a
	x_2)= (x_1 a) x_2$, for each $x_1, x_2 \in \A_0$ and $a \in \A$;
	
	\item[(iii)] an involution $*$ which extends the involution of $\A_0$
	is defined in $\A$ with the property $(ax)^*= x^*a^*$ and $(xa)^
	* =a^* x^*$ for each $x \in \A_0$ and $a \in \A$.
\end{itemize}

 We say that $(\A, \A_{\scriptscriptstyle 0})$ has an {\it unit} if there exists an
element $e\in \A_{\scriptscriptstyle 0}$ such that {$ae=ea=a$, for every $a \in \A$.}

A quasi *-algebra $(\A, \A_{\scriptscriptstyle 0})$ is said to
be {\it locally convex} if $\A$ is endowed with a topology $\tau$
which makes of $\A$ a locally convex space and such that the
involution $a \mapsto a^*$ and the multiplications {$\A\ni a \mapsto ax$,
$a \mapsto xa$, $x \in \A_{\scriptscriptstyle 0}$,} are continuous. If $\tau$ is a norm
topology and the involution is isometric with respect to the norm,
we say that $(\A[\|.\|],\Ao[\|.\|_0])$ is a {\it normed quasi *-algebra}   and,
if it is complete, we say it is a {\it Banach quasi*-algebra}, {\cite{aitbook}.}

For instance, let $\A_{\scriptscriptstyle 0}$ be a $C^\ast $-algebra, with norm $\|\cdot\|_0$ and
	involution $^*.$ Let $\| \cdot \|$
	be a norm on $\A_{\scriptscriptstyle 0}$, weaker than $\| \cdot \|_0$
	and such that, for every {$a,b \in  \A_0$}
	\begin{itemize}
		\item[(i)]  $\|ab\| \leq \|a\| \|b\|_0,$
		\item[(ii)]  {$ \|a^*\|=\|a\|.$}
	\end{itemize}
	Let $\A$ denote the $\| \cdot \|$-completion of $\A_{\scriptscriptstyle 0}$; then  the pair $(\A,\A_{\scriptscriptstyle 0})$ is called a (proper) {\em CQ*-algebra}, first introduced in \cite{fc}, which is a particular example of Banach quasi*-algebra.

As an example, the pair
$(L^p([0,1]), L^\infty([0,1])$, $1\leq p< +\infty$, may be regarded as an abelian
CQ*-algebra as well as a  Banach quasi *-algebra.

A second example is given by a non-commutative version of the previous one: let $\M$ be a von Neumann algebra and $\tau$ a normal faithful
semifinite trace defined on $\M_+$. For each $p\geq 1$, let
$${\mc J_p}=\{ X \in \M: \tau(|X|^p)<\infty \}.$$ Then ${\mc J_p}$ is a
*-ideal of $\M$. We denote with
$L^p(\tau)$ the Banach space completion of ${\mc J_p}$ with respect
to the norm
$$\|X\|_p := \tau(|X|^p)^{1/p}, \quad X \in {\mc J_p}.$$
We further call, as it is usually done in the literature, $L^\infty(\tau) = \M$. The pair $(L^p(\tau), L^\infty(\tau)\cap L^p(\tau))$  { is a non commutative Banach quasi *-algebra, \cite{cqellepi,nb}.} 

\section{Eigenvectors}\label{sect2}

Let $(\A[\|.\|],\Ao[\|.\|_0])$ be a Banach quasi *-algebra with identity $e\in\Ao$, and let $\varphi(.,.)$ be a sesquilinear form on it satisfying the following conditions:
\be\left\{
\begin{array}{ll}
\varphi(c,\alpha a+\beta b)=\alpha\varphi(c,a)+\beta\varphi(c,b),\\
\varphi(a,a)\geq0,\\
\varphi(ax,y)=\varphi(x,a^*y),
\end{array}
\right.
\label{21}\en
where $a, b, c\in\A$, $x,y\in\Ao$, and $\alpha,\beta\in\mathbb{C}$. Because of the positivity of $\varphi$, we automatically have
\be\left\{
\begin{array}{ll}
	\varphi(a,b)=\overline{\varphi(b,a)},\\
	|\varphi(a,b)|^2\leq\varphi(a,a)\varphi(b,b),
\end{array}
\right.
\label{22}\en
$\forall a,b\in\A$. We call $\Sc$ the set of sesquilinear forms satisfying (\ref{21}), and therefore (\ref{22}). It is clear that we can always assume that $\varphi(e,e)=1$. Indeed, if $\varphi(e,e)\neq1$, we can still define a new sesquilinear form, $\varphi_e$, as $\varphi_e(a,b)=\frac{\varphi(a,b)}{\varphi(e,e)}$, $a,b\in\A$, at least if $\varphi(e,e)\neq0$. Hence $\varphi_e(e,e)=1$, and $\varphi_e\in\Sc$. We refer to \cite{mariacamillo} for several details on sesquilinear forms on quasi *-algebras.

In what follows we will often make use of the following continuity condition:
\be
|\varphi(a,b)|\leq\gamma_\varphi\|a\|\|b\|
\label{23}\en
$\forall a,b\in\A$. Here $\gamma_\varphi$ is a strictly positive constant independent of $a$ and $b$. We call $\Sc_c$ the set of all the elements of $\Sc$ satisfying (\ref{23}), and we further introduce the set $\Sc_{c_0}=\{\varphi\in\Sc_c: \gamma_\varphi=1\}$.

\begin{prop}\label{prop0}
	Let $\varphi\in\Sc_c$, and $x\in\Ao$. Then calling
	\be
	\varphi_x(a,b)=\varphi(ax,bx),
	\label{20}\en
	$a,b\in\A$, we have that $\varphi_x\in\Sc_c$. In particular, if $\varphi\in\Sc_{c_0}$ and $\|x\|_0\leq1$, then $\varphi_x\in\Sc_{c_0}$.
\end{prop}
\begin{proof}
	We only proof the continuity of $\varphi_x$. For instance we have, using (\ref{23}) and the inequalities $\|ax\|\leq\|a\|\|x\|_0$ and $\|bx\|\leq\|b\|\|x\|_0$, 
	$$
	|\varphi_x(a,b)|=|\varphi(ax,bx)|\leq \gamma_\varphi\|ax\|\|bx\|=(\gamma_\varphi\|x\|_0^2)\|a\|\|b\|,
	$$
	for all $a,b\in\A$. It is clear that, in particular, $|\varphi_x(a,b)|\leq \|a\|\|b\|$ if $\gamma_\varphi=1$ and $\|x\|_0\leq1$.
	
\end{proof}

\begin{exe}\label{example}

If we consider the {{\it Banach quasi *-algebra}} $(L^p(\tau), L^\infty(\tau)\cap L^p(\tau)),$ it is not hard to produce examples of sesquilinear forms satisfying (\ref{21}) and (\ref{23}). For that we start introducing the set
	$${\mathcal{B}_+^p}=\{X\in L^{p/(p-2)}(\tau): \, X\geq0\, , \, \, \norm{X}_{p/(p-2)}\leq 1\} $$
	where $p>2$, with the agreement that $ p/{(p-2)}=\infty\,$ if $p=2.$
	
	For each $W \in {\mc B}_+^p$, we consider the right multiplication operator
	$$ R_{\scriptscriptstyle W}:L^p(\tau) \to L^{p'}(\tau); \hspace{8mm} R_{\scriptscriptstyle W} X= XW, \quad X \in L^p(\tau),$$ {where $\frac{1}{p}+\frac{1}{p'}=1$.}
	In this conditions for every  $p\geq2$ and 
	$W\in {\mathcal{B}_+^p}$, the sesquilinear form
	{$\vp(X,Y)=\tau[Y(R_{\scriptscriptstyle W}X){^\ast}] $} is an
	element of $\Sc$. In fact:

	\noindent for every $X\in L^p(\tau)$ we have $$\vp(X,X)=\tau
	[X(R_{\scriptscriptstyle W}X)^\ast]=\tau [X(XW)^{\ast}]=\tau
	[(XW)^{\ast}X]=\tau [W|X|^2]\geq 0.$$

	For every $X \in L^p(\tau)$,
	$A ,B\in$ { $L^\infty(\tau)\cap L^p(\tau)$} we get
	$$\vp(XA,B)=\tau (B(XAW)^{\ast})=\tau
	(BW^{\ast}A^{\ast}X^{\ast})=\tau(X^{\ast}BW^{\ast}A^{\ast})=\tau(X^*B(AW)^*)=\vp(A,X^
	{\ast}B).$$
	Finally, for every $ X,Y\in L^p(\tau)$,
	$$|\vp(X,Y)|\leq\norm{X}_p\norm{Y}_{p}\norm{W}_{p/p-2}\leq
	\norm{X}_p\norm{Y}_p.$$ 
	
	Then $\varphi\in\Sc_{c_0}$, in particular. Of course, if we relax condition $ \norm{W}_{p/(p-2)}\leq 1$, we can still check that $\varphi\in\Sc_c$.
\end{exe}

\begin{defn}\label{def1}
	Let $\varphi\in\Sc$ and let $a\in\A$. We say that $\varphi$ is an eigenstate of $a$ with eigenvalue $\lambda\in\mathbb{C}$, if 
	\be
	\varphi(b,a)=\lambda\varphi(b,e),
	\label{24}\en
	$\forall b\in\A$.
\end{defn}

It follows that $\varphi(e,a)=\lambda\varphi(e,e)=\lambda$, and therefore $\varphi(b,a)=\varphi(b,e)\varphi(e,a)$, $\forall b\in\A$. This could be seen as a sort of decomposition property of $\varphi$, which however only holds when $\varphi$ is an eigenstate of $a$, but not in general. 
Indeed we can easily check the following result:

\begin{lemma}\label{lemma1}
	Let $\varphi\in\Sc$ and let $a\in\A$. Then $\varphi$ is an {eigenstate of $a$} if and only if $\varphi(b,a)=\varphi(b,e)\varphi(e,a)$, $\forall b\in\A$.
\end{lemma}

Going back to Example \ref{example} {with $p=2$}, it is easy to construct an example of a sesquilinear form satisfying Definition \ref{def1}: {{\it in the Banach quasi *-algebra} $(L^2(\tau), L^\infty(\tau)\cap L^2(\tau)),$
if { $W\in  L^2(\tau)\cap {\mc B}_+^2$,} is a projection  then   $\vp(X,Y)=\tau[Y(R_{\scriptscriptstyle W}X){^\ast}]\in\Sc.$
For each $\alpha\in\mathbb{C}$  we now consider the operator $A=\alpha W$
$$\vp(X,A)=\tau[\alpha W(R_{\scriptscriptstyle W}X){^\ast}]= \tau[(\alpha R_{\scriptscriptstyle W}X){^\ast}]= { \alpha \varphi(X,e).}$$
Thus $\varphi$ is an eigenstate of $A$ with eigenvalue $\alpha.$

Next we have the following 

\begin{prop}\label{prop1}
	Let $\varphi\in\Sc$ and let $a\in\A$. Then $\varphi$ is an eigenstate of $a$ with eigenvalue $\lambda\in\mathbb{C}$ if and only if
	\be
	\varphi(a-\lambda e,a-\lambda e)=0.
	\label{25}\en
\end{prop}
\begin{proof}
	Let $\varphi$ be an eigenstate of $a$ with eigenvalue $\lambda\in\mathbb{C}$. Hence, from (\ref{24}), $\varphi(b,a-\lambda e)=0$, $\forall b\in\A$. In particular, this must be true if $b=a-\lambda e$, so that (\ref{25}) follows.
	
	Vice-versa, if (\ref{25}) holds true, (\ref{22}) implies that
	$$
	|\varphi(b,a-\lambda e)|\leq \varphi(b,b)\varphi(a-\lambda e,a-\lambda e)=0,
	$$
	so that $\varphi(b,a-\lambda e)=0$, $\forall b\in\A$. Hence $\varphi$ is an eigenstate of $a$ with eigenvalue $\lambda\in\mathbb{C}$.
	
\end{proof}

Suppose now that we have two different sesquilinear forms $\varphi_1,\varphi_2\in\Sc$ which are both eigenvectors of a given $a\in\A$, with the same eigenvalue $\lambda$. Then the new form {$\varphi(b,c)=q\varphi_1(b,c)+(1-q)\varphi_2(b,c)$,} $q\in[0,1]$ and $b,c\in\A$, is still an element of $\Sc$, and it is also an eigenvector of $a$ with the same eigenvalue $\lambda$. Moreover, if $\varphi_1,\varphi_2\in\Sc_c$, then $\varphi\in\Sc_c$ and, in particular, if $\varphi_1,\varphi_2\in\Sc_{c_0}$, then $\varphi\in\Sc_{c_0}$. Indeed, if $\varphi_1, \varphi_2$ both satisfy (\ref{21}), then $\varphi$ satisfies the same properties. Moreover, since $|\varphi_j(b,c)|\leq\gamma_{\varphi_j}\|b\|\|c\|,$ $j=1,2$, we have
$$
|\varphi(b,c)|\leq q |\varphi_1(b,c)|+(1-q)|\varphi_2(b,c)|\leq \left(q \gamma_{\varphi_1}+(1-q)\gamma_{\varphi_2}\right)\|b\|\|c\|,
$$
which is of the form (\ref{23}). If we further assume that $\varphi_1,\varphi_2\in\Sc_{c_0}$, then $\gamma_{\varphi_1}=\gamma_{\varphi_2}=1$, so that $q \gamma_{\varphi_1}+(1-q)\gamma_{\varphi_2}=1$. Hence $\varphi\in\Sc_{c_0}$. Finally, since both $\varphi_1$ and $\varphi_2$ satisfy (\ref{24}), we conclude that
$$
\varphi(b,a)=q\varphi_1(b,a)+(1-q)\varphi_2(b,a)=\lambda\left(q\varphi_1(b,e)+(1-q)\varphi_2(b,e)\right)=\lambda\varphi(b,e),
$$
which shows that, as stated, $\varphi$ is again an eigenstate of $a$ with eigenvalue $\lambda\in\mathbb{C}$. It is clear that these results can be extended to an arbitrary number of elements of, for instance, {$\Sc_{c_0}$:} if $\varphi_j\in\Sc_{c_0}$, $j=1,2,\ldots,N$, and if we have $\{q_j\in[0,1], \,j=1,2,\ldots,n\}$, with $\sum_{j=1}^nq_j=1$, then $\varphi=\sum_{j=1}^nq_j\varphi_j\in\Sc_{c_0}$. Moreover, if each $\varphi_j$ is an eigenstate of $a\in\A$ with eigenvalue $\lambda$, then $\varphi$ is an eigenstate of $a\in\A$ with eigenvalue $\lambda$ as well\footnote{Of course this latter properties does not depend on condition $\sum_{j=1}^nq_j=1$. }.

\begin{defn}\label{def2}
	Let us consider $\varphi_j\in\Sc$, $j=1,2,\ldots,n$. They are linearly independent if $\sum_{j=1}^n\alpha_j\varphi_j=0$ if and only if $\alpha_j=0$, $\forall j$.
\end{defn}
 It might be useful to recall that $\sum_{j=1}^n\alpha_j\varphi_j=0$ means that $\sum_{j=1}^n\alpha_j\varphi_j(a,b)=0$, $\forall a,b\in\A$. Following \cite{denittis} we prove the following:
 
 \begin{prop}\label{prop2}
 Let us consider $\varphi_j\in\Sc$, $j=1,2,\ldots,n$, different eigenstates of $a\in\A$, with $n$ different eigenvalues $\lambda_j$. Then these forms are linearly independent.
 \end{prop}
The proof is not very different from the one in \cite{denittis} for states on C*-algebras and will not be given here. 

Our Definition \ref{def1} implies the following {\em natural} result: 

\begin{lemma}\label{lemma2}
If $a=a^*$, then $\lambda\in\mathbb{R}$ in (\ref{24}). 
\end{lemma}
\begin{proof}
Indeed we have
$$
\lambda=\varphi(e,a)=\varphi(a^*,e)=\varphi(a,e)=\overline{\varphi(e,a)}=\overline{\lambda},
$$
 so that the reality of $\lambda$ follows.
 
\end{proof}

If $a^*\neq a$, then $\lambda$ is not necessarily real. In this general case it is possible to proceed as follows.

\begin{prop}\label{prop3}
	Let $\varphi\in\Sc$. Then, defining
	\be
	\psi(a,b)=\varphi(b^*,a^*),
	\label{26}\en
	$a,b\in\A$, $\psi$ is also a sesquilinear positive form. Moreover 
	\be
	\psi(xa,y)=\psi(x,ya^*),
	\label{27}\en
	$\forall a\in\A$, $x,y\in\Ao$.	
	Furthermore, if $\varphi$ satisfies (\ref{23}), $\psi$ satisfies also
	\be
	|\psi(a,b)|\leq\gamma_\varphi\|a\|\|b\|,
	\label{28}\en
	$a,b\in\A$.
	
\end{prop} 
\begin{proof}
	The proof that $\psi$ is sesquilinear and positive is trivial and will not be given here. As for (\ref{27}) we have, using (\ref{26}) and (\ref{21})$_3$ 
	$$
		\psi(xa,y)=\varphi(y^*,(xa)^*)=\varphi(y^*,a^*x^*)=\varphi(ay^*,x^*)=\psi(x,(ay^*)^*)=\psi(x,ya^*).
	$$
	Because of (\ref{23}) and (\ref{26}) we also have
	$$
	|\psi(a,b)|=|\varphi(b^*,a^*)|\leq \gamma_\varphi\|b^*\|\|a^*\|,
	$$
	so that (\ref{28}) follows.
	
\end{proof}

We notice that the constant appearing in (\ref{23}) and (\ref{28}) is the same. Notice also that, even if $\varphi\in\Sc$, the form $\psi$ does not belong to $\Sc$, if $(\A,\Ao)$ is not abelian, since (\ref{21})$_3$ is replaced here by (\ref{27}). In this case we could introduce a set $\tilde\Sc$, with (\ref{27}) rather than (\ref{21})$_3$ and the other two in (\ref{21}) unchanged. Then $\psi\in\tilde\Sc$ if and only if $\varphi\in\Sc$.

\begin{prop}\label{prop4}
	Let $\varphi\in\Sc$, $a\in\A$, $\lambda\in\mathbb{C}$ and $\psi$ the form in (\ref{26}). The following statements are equivalent:
	\be
	\begin{array}{ll}
		(i) \qquad \varphi(b,a)=\lambda\,\varphi(b,e),\\
		(ii) \qquad \psi(b,a^*)=\overline{\lambda}\,\psi(b,e),\\
		(iii) \qquad \varphi(a,b)=\overline{\lambda}\,\varphi(e,b),\\
		(iv) \qquad \psi(a^*,b)={\lambda}\,\psi(e,b),
	\end{array}
	\label{29}\en
	$\forall b\in\A$.
	Moreover, if any of the above equalities is satisfied, then
	\be
	\varphi(a,a)=\psi(a,a)=|\lambda|^2.
	\label{210}\en
\end{prop} 
\begin{proof}
	We only show here that $(i)$ implies $(ii)$. The proofs of all these claims are similar, and will not be repeated. Let us assume that $(i)$ is satisfied: $\varphi$ is an eigenstate of $a$ with eigenvalue $\lambda$. Then we have
	$$
	\psi(b,a^*)=\varphi(a,b^*)=\overline{\varphi(b^*,a)}=\overline{\lambda\,\varphi(b^*,e)}=\overline{\lambda}\varphi(e,b^*)=\overline{\lambda}\psi(b,e),
	$$
	because of $(i)$ and (\ref{26}). The others implications can be proved in a similar way. 
	
	As for (\ref{210}), we have for instance $\varphi(a,a)=\lambda\varphi(a,e)$, using $(i)$. But, using $(iii)$ with $b=e$ we also have $\varphi(a,e)=\overline{\lambda}\varphi(e,e)=\overline{\lambda}$, so that $\varphi(a,a)=|\lambda|^2$. Similarly we can prove that $\psi(a,a)=|\lambda|^2$.

\end{proof}

This Proposition shows, in particular, that, if $\varphi$ is an eigenstate of $a$ with eigenvalue $\lambda$, then $\psi$ is an eigenstate of $a^*$ with eigenvalue $\overline\lambda$. Incidentally, formulas $(iii)$ and $(iv)$ suggest to introduce the notion of {\em left} and {\em right} eigenvectors, depending on the position of $a$ inside the sesquilinear forms $\varphi$ or $\psi$. For instance, from $(i)$ we could say that  $\varphi$ is a right eigenstate of $a$ with eigenvalue $\lambda$, while from $(iii)$ $\varphi$ can be called a left eigenstate of $a$ with eigenvalue $\overline\lambda$. Similarly, see $(ii)$ and $(iv)$, $\psi$  is a right eigenstate of $a^*$ with eigenvalue $\overline\lambda$ and a left eigenstate of $a^*$ with eigenvalue $\lambda$. 

It is possible to introduce a norm on our sesquilinear forms in $\Sc_{c}$. Indeed we can put
\be
\|\varphi\|=\sup_{\|a\|=1}\varphi(a,a)=\sup_{\|a\|=\|b\|=1}|\varphi(a,b)|,
\label{211}\en
where the equivalence of the two expressions follow from (\ref{22})$_2$. It is clear that, $\forall\varphi\in\Sc_{c_0}$, $\|\varphi\|=1$. In fact we have, recalling that our quasi *-algebra has an identity $e$, $\|\varphi\|=\sup_{\|a\|=1}\varphi(a,a)\leq \sup_{\|a\|=1}\|a\|^2=1$, and $\|\varphi\|=\sup_{\|a\|=1}\varphi(a,a)\geq \varphi(e,e)=1$.

 We conclude this section by noticing that the results of Proposition \ref{prop4} can be enriched if $\varphi$ is an eigenstate of an element $a_0$ in $\Ao$, $a_0\in\Ao$.  Hence we have, for instance, $\varphi(b,a_0)=\lambda\,\varphi(b,e)$, $\forall b\in\A$. Then we also have $\varphi(b,a_0^n)=\lambda^n\,\varphi(b,e)$, $\forall b\in\A$, $n\geq0$. In fact, first of all we notice that, since $a_0\in\Ao$, its powers are all well defined in $\Ao$. Then we have, in particular
 $$
 \varphi(b,a_0^2)=\varphi(a_0^*b,a_0)=\lambda\varphi(a_0^*b,e)=\lambda\varphi(b,a_0)=\lambda^2\varphi(b,e),
 $$ 
 $\forall b\in\A$. And so on. Of course, a similar result can be extended for all polynomials: let $p(a_0)$ be a polynomial in $a_0$, then $\varphi(b,p(a_0))=p(\lambda)\varphi(b,e)$. Similar results can be deduced out of the other equalities in (\ref{29}).

\section{Ladder elements}\label{sect3}

In this section we will discuss in some details an application of what we have done in Section \ref{sect2}, and in particular we will use ladder operators of a specific kind to construct a family of sesquilinear forms which are eigenvectors of a single element of $\Ao$, with different eigenvalues. More explicitly, we will use elements of $\Ao$ satisfying a specific commutation rule which is motivated by the so-called {\em quons}, which are operators acting on some Hilbert space and which depend on a real parameter $q\in[-1,1]$, see \cite{moh,fiv,gree}. Few facts on these operators are listed in the Appendix, where it is also discussed that they are bounded operators, if $q\in[-1,1[$.

With this in mind we consider here two elements $x_0, y_0\in\Ao$ satisfying the following {\em $q$-mutator}:
\be
[x_0,y_0]_q=x_0y_0-qy_0x_0=e.
\label{31}\en
We call $n_0=y_0x_0$, which is still an element of $\Ao$. It is clear that, if $x_0=y_0^*$, then $n_0=n_0^*$. We assume that an element $\varphi_0\in\Sc_c$ exists such that
\be
\varphi_0(b,x_0)=0,
\label{32}\en
$\forall b\in\A$. This means that $\varphi_0$ is an eigenstate of $x_0$ with eigenvalue $\lambda_0=0$. Notice that (\ref{32}) implies that $\varphi_0(b,c_0x_0)=0$ for all $b\in\A$ and $c_0\in\Ao$. Indeed we have $\varphi_0(b,c_0x_0)=\varphi_0(c_0^*b,x_0)=0$ since $c_0^*b\in\A$. Further, the continuity of {$\varphi_0$} implies that  $\varphi_0(b,cx_0)=0$ for all $b,c\in\A$. This is because it surely exists a sequence $\{c_n\in\Ao\}$ such that $\|c-c_n\|\rightarrow0$. Then we also have {$\|cx_0-c_nx_0\|\rightarrow0$,} so that $\varphi_0(b,cx_0)=\lim_n\varphi_0(b,c_nx_0)=0$. Summarizing:
\be
\varphi_0(b,cx_0)=0,
\label{32bis}\en
$\forall b,c\in\A$.

\vspace{2mm}

{\bf Remark:--} The existence of the sesquilinear form $\varphi_0$ satisfying (\ref{32}) is not a trivial requirement. Indeed it is not so different from what is done in the literature for bosons, or for pseudo-bosons: in both cases, one asks for the existence of a {\em vacuum state}. But, while this state exists, and belongs to the Hilbert space (e.g., to $\ltwo$) in the bosonic case, this existence is not guaranteed for pseudo-bosons\footnote{In what we are doing here, the bosonic case corresponds to the case $y_0=x_0^*$, while the pseudo-bosonic (or pseudo-quonic, to be more precise) to $y_0\neq x_0^*$.}. In this case, even when the vacuum exists, it could be an element not in $\ltwo$. We refer to \cite{bagspringer} where the role of the distributions is discussed in details.

\vspace{2mm}

The following result holds true:
\begin{prop}\label{prop5}
	Calling
	\be
	\varphi_l(a,b)=\varphi_{l-1}(ay_0,by_0),
	\label{33}\en
	$a,b\in\A$, $l=1,2,3,\ldots$, then 
	\be\varphi_l(a,n_0)=\beta_l\varphi_l(b,e),
	\label{34}\en
	$\forall b\in\A$, where
	\be
	\beta_l =\left\{
	\begin{array}{ll}
		0, \hspace{2.6cm} \mbox{if } l=0,\\
		1+q\beta_{l-1}, \hspace{1.2cm} \mbox{if } l\geq1.
	\end{array}
	\right.
	\label{35}
	\en	
\end{prop}
\begin{proof}	
	We prove our claim by induction: the statement is true for $l=0$ because of (\ref{32}). Let us now suppose that $\varphi_{l-1}(a,n_0)=\beta_{l-1}\varphi_{l-1}(b,e)$ for a fixed $l$. We want to deduce then that $\varphi_{l}(a,n_0)=\beta_{l}\varphi_{l}(b,e)$. Indeed we have, since $n_0y_0=y_0x_0y_0=y_0(e+qy_0x_0)=y_0+qy_0^2x_0$,
	$$
	\varphi_{l}(a,n_0)=\varphi_{l-1}(ay_0,n_0y_0)=\varphi_{l-1}(ay_0,y_0)+q\varphi_{l-1}(ay_0,y_0^2x_0)=\varphi_{l}(a,y)+q\varphi_{l-1}(ay_0,y_0^2x_0).
	$$
	Now, because of the induction hypothesis and of (\ref{33}), we have $$\varphi_{l-1}(ay_0,y_0^2x_0)=\varphi_{l-1}(y_0^*ay_0,n_0)=\beta_{l-1}\varphi_{l-1}(y_0^*ay_0,e)=\beta_{l-1}\varphi_{l-1}(ay_0,y_0)=\beta_{l-1}\varphi_{l}(a,e),$$
	which implies that 
	$$
		\varphi_{l}(a,n_0)=(1+q\beta_{l-1})\varphi_{l}(a,e),
	$$
	as we had to prove.

\end{proof}

Incidentally we observe that, if $-1<q<1$, then $\beta_l=1+q+q^2+\cdots+q^l=\frac{1-q^{l+1}}{1-q}$, $\forall l\geq1$.

What this proposition shows is that, for elements of $\Ao$ satisfying (\ref{31}), if we know the vacuum of the operator $x_0$, we can construct a full sequence of forms as in (\ref{33}), which are all eigenvectors of a single element $n_0$, see (\ref{34}). In particular, since $y_0$ allows us to move from $\varphi_{l-1}$ to $\varphi_l$, we will call $y_0$ a {\em raising operator}. Similarly, $x_0$ is a {\em lowering operator}, as Proposition \ref{prop7} below shows.

\begin{prop}\label{prop7}
For all $a,b\in\A$ we have
	\be
	\varphi_l(ax_0,bx_0)=\beta_l^2\varphi_{l-1}(a,b),
	\label{36}\en	
	 $\forall l\geq1$.
\end{prop}
\begin{proof}
We start recalling that $x_0y_0=e+qy_0x_0=e+qn_0$.

Now, let us suppose first that $a\in\Ao$. Hence we have, because of (\ref{33}),
$$
\varphi_l(ax_0,bx_0)=\varphi_{l-1}(ax_0y_0,bx_0y_0)=\varphi_{l-1}(a(e+qn_0),b(e+qn_0))=
$$
$$
=\varphi_{l-1}(a,b)+q\varphi_{l-1}(a,bn_0)+q\varphi_{l-1}(an_0,b)+q^2\varphi_{l-1}(an_0,bn_0).
$$
Now we have
$$
\varphi_{l-1}(a,bn_0)=\varphi_{l-1}(b^*a,n_0)=\beta_{l-1}\varphi_{l-1}(b^*a,e)=\beta_{l-1}\varphi_{l-1}(a,b),
$$
because of (\ref{34}) and recalling that, under our assumption on $a$, $b^*a\in\A$ is well defined. Similarly we deduce that $\varphi_{l-1}(an_0,b)=\beta_{l-1}\varphi_{l-1}(a,b)$ and $\varphi_{l-1}(an_0,bn_0)=\beta_{l-1}^2\varphi_{l-1}(a,b)$. Hence
$$
\varphi_l(ax_0,bx_0)=\left(1+q\beta_{l-1}\right)^2\varphi_{l-1}(a,b),
$$
which is exactly the (\ref{36}). 

If we now take $a\in\A$, then there exists a sequence $\{a_n\in\Ao\}$ such that $\|a-a_n\|\rightarrow0$ which, in turns, implies that $\|ax_0-a_nx_0\|\rightarrow0$. Then, using the continuity of each $\varphi_k$,
$$
\varphi_{l}(ax_0,bx_0)=\lim_n\varphi_{l}(a_nx_0,bx_0)=\beta_l^2\lim_n \varphi_{l-1}(a_n,b)=\beta_l^2\varphi_{l-1}(a,b),
$$
$a,b\in\A$.

\end{proof}

It is interesting to stress that Proposition \ref{prop2} and Proposition \ref{prop5}, together, imply that the various $\varphi_l$ are linearly independent.

\subsection{Towards coherent forms}

In the literature ladder operators acting on Hilbert spaces are usually associated to particular vectors depending on a variable $z\in\mathbb{C}$ which are called coherent states, and which, among other properties, are eigenstates of the relevant lowering operator one is considering. Many mathematical results and several physical applications exist for coherent states, and for their possible generalizations. We refer to \cite{klauder}-\cite{bagspringer}
for a first reading.  

We will now show that ladder elements of the kind considered in this section can also be used to define what we will call {\em almost coherent forms}, i.e. forms which are close to be eigenstates of the lowering operator $x_0$ in (\ref{31}). Our procedure reflects the one usually adopted in Hilbert spaces. We define (formally, for the moment), the following form:
\be
\Phi_z(a,b)=\sum_{l=0}^\infty\,\frac{z^l}{(\beta_l!)^2}\,\varphi_l(a,b),
\label{37}\en
$\forall a,b\in\A$ and for $z\in\D\subseteq\mathbb{C}$, a subset of $\mathbb{C}$ to be defined. In other words, $\D$ is the set of all the complex $z$'s for which the series in (\ref{37}) does converge for all $a$ and $b$ in $\A$. Notice that, in principle, $\D$ could consist of the single point $z=0$. However, we will show that this is not the case. In (\ref{37}) we have $\beta_0!=1$ and $\beta_l!=\beta_1\beta_2\cdots\beta_l$, $l\geq1$. 

To study the convergence of $\Phi_z(a,b)$ we start stressing that, since $\varphi_0\in\Sc_c$, each $\varphi_l$ belongs to $\Sc_c$ as well, and in particular
\be
|\varphi_l(a,b)|\leq \gamma_{\varphi_0}\|y_0\|_0^{2l}\|a\|\|b\|,
\label{38}\en
$a,b\in\A$, $l=0,1,2,3,\ldots$. This is exactly of the kind (\ref{23}), with $\gamma_\varphi=\gamma_{\varphi_0}\|y_0\|_0^{2l}$. 

Then, going back to (\ref{37}) we  have
$$
|\Phi_z(a,b)|\leq \sum_{l=0}^\infty\,\frac{|z|^l}{(\beta_l!)^2}\,|\varphi_l(c,b)|\leq \gamma_{\varphi_0}\|a\|\|b\|\sum_{l=0}^\infty\,\frac{(|z|\|y_0\|_0^{2})^l}{(\beta_l!)^2},
$$
which is clearly a power series with radius of convergence 
\be\rho'=\lim_{l}\beta_{l+1}^2=\left\{
\begin{array}{ll}
	\infty, \hspace{2.6cm} \mbox{if } q=1,\\
	\frac{1}{1-q}, \hspace{2.4cm} \mbox{if } q\in]-1,1[.
\end{array}
\right.
\label{39}\en
Hence the series (\ref{37}) exists for $|z|<\rho=\frac{\rho'}{\|y_0\|_0^{2}}$, that is for all $z\in C_\rho(0)$, the circle centered in the origin of the complex plane and of radius $\rho$, which is exactly the set $\D$ above we were trying to identify. It is now a simple exercise to check that 
\be
\Phi_z(ax_0,bx_0)=z\Phi_z(a,b),
\label{310}\en
$\forall z\in C_\rho(0)$ and for all $a,b\in\A$. Of course, this is not as saying that $\Phi_z(a,b)$ is an eigenstate of $x_0$ with eigenvalue $z$, but it is not so far away, and in any case it is an interesting feature of the form $\Phi_z$. We hope to be able to define proper coherent forms soon, and to study their properties, including the possibility of deducing some sort of resolution of the identity, as one usually expects from ordinary coherent states, \cite{gazeaubook}. 


\subsection{Back to the $\varphi_l$'s}

In \cite{denittis} the authors prove, in  a different context with respect to the one we are considering here, that the states they call egenvectors  become, after constructing their GNS\footnote{GNS stands for e Gelfand-Naimark-Segal} representations, {\em ordinary}\footnote{i.e., in a purely Hilbertian sense.} eigenvectors of the operator which represents the original element of the C*-algebra in the Hilbert space arising from the GNS construction. The same result can be proven in our framework, and this will have interesting consequences, as we will see. To start our analysis, it is worth to  recall briefly how the GNS construction works in the present case.

{Given a {Banach} quasi $\ast$-algebra $(\A [\| \cdot \|], \A_0[\| \cdot \|_0])$ with identity $e\in\A_0$ and $\varphi\in \Sc$,
we put  
{$$ \mathfrak{N}_\varphi =\{ a\in\mathcal{A}: \; \varphi(a,a)=0\}=\{ a\in\mathcal{A}: \; \varphi(a,b)=0, \; \forall b\in \A\}.$$

	Let 	{$\D_\varphi^0:=\A_0/ \N_\vp=\{ \lambda_\varphi(x):= x+\mathfrak{N}_\varphi: \; x\in\A_0\}$} and { $\D_\varphi:=\A/ \N_\vp=\{ \lambda_\varphi(a):= a+\mathfrak{N}_\varphi: \; a\in\A\}$ }  be the usual quotient sets.

	For any $a\in\A$ we put
$$ \pi_\varphi(a)\lambda_\varphi(x)=\lambda_\varphi(ax), \; x\in\A_0.
$$ Then by the third condition in (2.1) $\pi_\varphi(a)$ is an element of { $\mathcal{L}^\dagger(\D_\varphi^0, \D_\varphi)$} (the set of all linear mappings $X$ { from $\D_\varphi^0$ into $ \D  _\varphi$)},  and $\pi_\varphi(a)^\dagger := \pi_\varphi(a)^\ast\lceil_{\D_\varphi}=\pi_\varphi(a^\ast)$ and $\pi_\varphi(ax)=\pi_\varphi(a)\pi_\varphi(x)$ for all $a\in\A$ and $x\in\A_0$.
 
Moreover if $\xi_\varphi:=\lambda_\varphi(e)$ we define in $\pi_\varphi(\A)\xi_\varphi$ a inner product by 

\be
<\pi_\varphi(a)\xi_\varphi,\pi_\varphi(b)\xi_\varphi>=\varphi(a,b).
\label{311}\en
 We denote by  $\Hil_\varphi$ the completion of $\pi_\varphi(\A)\xi_\varphi.$ 
Then, $\xi_\varphi$ is a cyclic vector for $\pi_\varphi$, that is, $\pi_\varphi(\A)\xi_\varphi$ is dense in $\Hil_\varphi$.   We call $(\Hil_\varphi,\xi_\varphi,\pi_\varphi)$ the GNS-representation of the Banach quasi $\ast$-algebra $(\A[\|\cdot\|],\A_0 [\|\cdot\|_0])$ for $\varphi$.

With this in mind, it is possible to prove the following.
\begin{prop}\label{prop8}
	Let $\varphi\in\Sc_{c_0}$ be an eigenstate of $a\in\A$ with eigenvalue $\lambda$. Then $\xi_\varphi$ satisfies $\pi_\varphi(a)\xi_\varphi=\lambda\xi_\varphi$, and vice-versa. 
\end{prop}
\begin{proof}
Indeed we have, using Proposition \ref{prop1}, that $\varphi(b,a)=\lambda\varphi(b,e)$ for all $b\in\A$, if and only if $\varphi(a-\lambda e,a-\lambda e)=0$. But this, using (\ref{311}), is equivalent to
$$
0=\langle\pi_\varphi(a-\lambda e)\xi_\varphi,\pi_\varphi(a-\lambda e)\xi_\varphi\rangle=\|\pi_\varphi(a-\lambda e)\xi_\varphi\|^2,
$$
	so that our claim easily follows.

\end{proof}

With this in mind we consider now what happens for our family of forms $\varphi_l$ in (\ref{32}) and (\ref{34}). In particular we will show that all these forms produce a single *-representation (and therefore a single Hilbert space $\Hil$ and a single cyclic vector) and, even more interesting (and expected!) that it is natural to associate the various $\varphi_l$ to different vectors of $\Hil$ which, under natural hypotheses on $x_0$ and $y_0$, are orthogonal.

Let us consider first the form $\varphi_0$ in (\ref{32}), and let us assume, to simplify the notation that $\varphi_0\in\Sc_{c_0}$. This is not a big requirement: if $\gamma_{\varphi_0}\neq1$, we could replace $\varphi_0$ with $\varphi_0'=\frac{1}{\gamma_{\varphi_0}}\varphi_0$.
Using (\ref{311}), there is a triple $(\Hil_0,\xi_0,\pi_0)$ such that
\be
\varphi_0(a,b)=\langle\pi_0(a)\xi_0,\pi_0(b)\xi_0\rangle_0,
\label{312}\en
$\forall a,b\in\A$ and where $\langle\cdot,\cdot\rangle_0$ is the scalar product in $\Hil_0$. Now, what is interesting for us is that $\varphi_1$ does not necessarily produces another triple $(\Hil_1,\xi_1,\pi_1)$, as we will now show. The reason is in formula (\ref{34}), which shows that there is a relation between $\varphi_1$ and $\varphi_0$:
$$
\varphi_1(a,b)=\varphi_0(ay_0,by_0)=\langle\pi_0(ay_0)\xi_0,\pi_0(by_0)\xi_0\rangle_0=\langle\pi_0(a)\pi_0(y_0)\xi_0,\pi_0(b)\pi_0(y_0)\xi_0\rangle_0=
$$
$$
=\langle\pi_0(a)\xi_1,\pi_0(b)\xi_1\rangle_0,
$$
where we have defined $\xi_1=\pi_0(y_0)\xi_0\in\Hil_0$. Similarly we have
$$
\varphi_2(a,b)=\varphi_0(ay_0^2,by_0^3)=\langle\pi_0(ay_0^2)\xi_0,\pi_0(by_0^2)\xi_0\rangle_0=\langle\pi_0(a)\xi_2,\pi_0(b)\xi_2\rangle_0,
$$
where {$\xi_2=\pi_0(y_0^2)\xi_0=\pi_0(y_0)^2\xi_0\in\Hil_0$.} And so on. This shows that the triple $(\Hil_0,\xi_0,\pi_0)$ is sufficient to deal with all the forms $\varphi_l$, $l=0,1,2,3,\ldots$, and that, in general
\be
\varphi_l(a,b)=\langle\pi_0(a)\xi_l,\pi_0(b)\xi_l\rangle_0,
\label{313}\en
$\forall a,b\in\A$, $l=0,1,2,3,\ldots$. Here $\xi_l=\pi_0(y_0)^l\xi_0\in\Hil_0$. It is interesting to observe now that the various $\xi_l$ satisfy an orthogonality condition. To show this aspect we start rewriting
\be
 \langle\xi_k,\xi_l\rangle_0=\langle\pi_0(y_0)^k\xi_0,\pi_0(y_0)^l\xi_0\rangle_0=\varphi_0(y_0^k,y_0^l)
\label{314}\en
where we have used, in particular, (\ref{312}), and the fact that $\pi_0(y_0)^n=\pi_0(y_0^n)$, for all $n\geq0$. Hence, for instance,
$
 \langle\xi_1,\xi_0\rangle_0=\varphi_0(y_0,e)=\varphi_0(e,y_0^*)
$. Now, if $y_0^*=x_0$, we have $\varphi_0(e,y_0^*)=\varphi_0(e,x_0)=0$ because of (\ref{32}). Therefore $\langle\xi_1,\xi_0\rangle_0=0$. Similarly we could check that $\langle\xi_l,\xi_0\rangle_0=0$, $\forall l\geq1$, or that $\langle\xi_2,\xi_1\rangle_0=0$. In this case we have, assuming again that $y_0=x_0^*$ and recalling that $x_0y_0=e+qy_0x_0$,
$$
\langle\xi_2,\xi_1\rangle_0=\varphi_0(y_0^2,y_0)=\varphi_0(y_0,x_0y_0)=\varphi_0(y_0,e)+q\varphi_0(y_0,y_0x_0)=
$$
$$
=\varphi_0(e,x_0)+q\varphi_0(x_0y_0,x_0)=0.
$$
These results can be extended:
\begin{prop}\label{prop9}
	If $y_0=x_0^*$ then 
	\be
	\langle\xi_k,\xi_l\rangle_0=0
	\label{315}\en
	for all $k\neq l$.
\end{prop}
\begin{proof}
	It is enough to observe that, if $y_0=x_0^*$, then $n_0=n_0^*$. Then, because of Proposition \ref{prop8} and of (\ref{34}), we have $\pi_0(n_0)\xi_l=\beta_l\xi_l$. Then our claim follows with standard steps from the facts that $\pi_0$ is a *-representation, and that $\beta_l\neq\beta_k$ whenever $l\neq k$.

\end{proof}

\vspace{2mm}

{\bf Remarks:--} (1) This proposition suggests a possible way to define a notion of orthogonality of the forms $\varphi_l$: we could say that these are {\em orthogonal} because their corresponding vectors $\xi_l$ are orthogonal in $\Hil_0$. However, this definition is not really easy to extend to arbitrary forms $\Omega$ and $\tilde \Omega$, since there is no reason a priori which guarantees that they produce a single Hilbert space when constructing their related GNS-like representations, which was a key feature in our construction.

(2) The condition $y_0=x_0^*$ is rather important, since it implies that $n_0=y_0x_0=n_0^*$, which is used in the proof of Proposition \ref{prop9}. If $y_0\neq x_0^*$, then there is no reason to expect that the orthogonality between the various $\xi_l$ has to be true. However, in this case, we could rather imagine that a second family $\eta_k\in\Hil_0$ can be defined, which is biorthogonal to the $\xi_l$: 	$\langle\xi_l,\eta_k\rangle_0=0$, if $k\neq l$. This is a rather common feature in non-Hermitian quantum mechanics, \cite{bagspringer}-\cite{bender}, and it is exactly what we will briefly consider below.

(3) It remains open the possibility of using the norm of sesquilinear forms to check their mutual orthogonality, as proposed in a different context in \cite{pede,denittis}. This approach does not look so natural for us, but a deeper understanding of this possibility is part of our future work.

\vspace{2mm}

In Proposition \ref{prop9} we have assumed that $y_0=x_0^*$. Let us briefly consider what happens when $y_0\neq x_0^*$. This is interesting, and can be seen as going from, say, bosons (or fermions, or quons) to pseudo-bosons (or pseudo-fermions, or pseudo-quons), \cite{bagspringer}.

In this case, in complete analogy with what we have done in (\ref{32}) and after, we assume that a second element $\eta_0\in\Sc_c$ exists which is an eigenstate of $y_0^*$ with zero eigenvalue:

\be
\eta_0(b,y_0^*)=0,
\label{316}\en
$\forall b\in\A$. It is clear that, if $x_0=y_0^*$, then $\eta_0$ and $\varphi_0$ collapse, while in general they are expected to be different forms. Then, calling
	\be
	\eta_l(a,b)=\varphi_{l-1}(ax_0^*,bx_0^*),
	\label{317}\en
	$a,b\in\A$, $l=1,2,3,\ldots$, we have 
	\be\eta_l(a,n_0^*)=\beta_l\eta_l(b,e),
	\label{318}\en
	$\forall b\in\A$. The counterpart of (\ref{36}) is now
	\be
	\eta_l(ay_0^*,by_0^*)=\beta_l^2\eta_{l-1}(a,b),
	\label{319}\en	
	$\forall l\geq1$ and $a,b\in\A$. We can consider a GNS-like construction for each $\eta_l$. However, they are all  related to a single GNS, the one we construct out of $\eta_0$, $\pi^{(\eta)}_0$, analogously to (\ref{313}):
	\be
	\eta_0(a,b)=\langle\pi^{(\eta)}_0(a)\nu_0,\pi^{(\eta)}_0(b)\nu_0\rangle_{\eta_0},
	\label{320}\en
$a,b\in\A$, and the scalar product is the one in $\Hil_{\eta_0}$. Moreover, $\nu_l=\pi^{(\eta)}_0({x_0^*})^l\nu_0$. It is now a trivial exercise to show that, if $(\Hil_0,\xi_0,\pi_0)=(\Hil_{\eta_0},\nu_0,\pi^{(\eta)}_0)$, then
\be
	\langle\xi_k,\nu_l\rangle=0,
\label{321}\en
for all $k\neq l$. Here $\langle\cdot,\cdot\rangle$ is the scalar product of $\Hil_0=\Hil_{\eta_0}$. Hence the two sets of vectors are biorthogonal, and then we say that also the sesquilinear forms they are derived from are biorthogonal. However, it should be stressed that requiring  that $(\Hil_0,\xi_0,\pi_0)=(\Hil_{\eta_0},\nu_0,\pi^{(\eta)}_0)$ is a strong assumption, and it is not guaranteed at all, in general.

\section{Conclusions}\label{sect4}

In this paper we have proposed a possible extension of the notion of eigenvectors in the context of Banach quasi *-algebra, and we have deduced several properties of these specific sesquilinear forms, our eigenstates. An interesting result is that, if $\varphi$ is an eigenstate of a certain $a\in\A$ with eigenvalue $\lambda$, it is possible to define another form, $\psi$, as in (\ref{29}), which is eigenstate of $a^*$ with eigenvalue $\overline{\lambda}$.

As an explicit application, we have constructed a family of eigenstates of a number-like operator defined in terms of certain ladder elements which obey suitable commutation relations of the q-uonic type. This particular application allows us to introduce the notion of orthogonality and of biorthogonality for our sesquilinear forms, by using the GNS representation defined by these forms.

Many other aspects are still open and interesting in our research, both from a mathematical side and for possible physical applications. The possibility of defining coherent forms is just one of these: we have already done some steps here, see (\ref{37}), but there is still a long way to go. The orthogonality, and the biorthogonality, at the level of forms, rather than in representation, is also a quite intriguing aspect. And possible physical applications, possibly to ladder elements of different nature, look interesting. So, our analysis is in progress.

\section*{Acknowledgements}

F. B. and S. T. acknowledge partial financial support from Palermo University, via FFR2021 "Bagarello" and "Triolo", and from G.N.F.M.  and G.N.A.M.P.A. of the INdAM. H. I. acknowledges partial financial support from {the Università degli Studi di Palermo, by the program  {\em CORI-2021-D3-D26-160891-Azione D3}.}

\section*{Data accessibility statement}

This work does not have any experimental data.

%
%
%
%
%

\section*{Funding statement}

F. B. and S. T. acknowledge partial financial support from Palermo University, via FFR2021 "Bagarello" and "Triolo", and from G.N.F.M.  and G.N.A.M.P.A. of the INdAM. H. I. acknowledges partial financial support from {the Università degli Studi di Palermo, by the program  {\em CORI-2021-D3-D26-160891-Azione D3}.}

\renewcommand{\theequation}{A.\arabic{equation}}

\section*{Appendix A: quons in $\Hil$}\label{appendixA}

This appendix is devoted to a
short review on the standard results on quons, which are taken mainly from Refs.
\cite{moh,fiv,gree}.

The standard commutation rule for a single mode quon is
\be
a_qa_q^\dagger-qa_q^\dagger a_q=\1, \qquad  q\in[-1,1].
\label{a1}
\en
We see that this relation interpolates between bosons ($q=1$) and fermions ($q=-1$). 

If $\Phi_0$ is the vector of the Hilbert
space $\Hil$ annihiled by the annihilation quon operator $a_q$, $a_q \Phi_0=0$,
then the set of the  vectors defined recursively by 
\be
\Phi_{n+1}=\frac{1}{\gamma_n}a_q^\dagger\Phi_{n},  \label{a2}
\en
is an orthonormal basis for $\Hil$. The value of the normalization constant
$\gamma_n$ depends on $q$ and $n$ through the expression 
\be
\gamma_n^2\! =\left\{
\begin{array}{ll}
	\frac{1-q^{n+1}}{1-q}, \hspace{3cm} \mbox{if } q\neq 1,\\
	n+1, \hspace{3.2cm} \mbox{if } q=1.
\end{array}
\right.
\label{a3}
\en
Defining the self-adjoint operator
$
N_q=a_q^\dagger a_q,
$
it is easy to see that
$$
N_q \Phi_{n}= \gamma_{n-1}^2\Phi_{n},
$$
where $\gamma_{-1}:=0$.

The operator $a_q$, is bounded for any $q$ in $[-1,1[$. Indeed we have
$$
\|a_q\|^2=\sup_{\Psi\in\Hil: \|\Psi\|\leq1}<a_q\Psi,a_q\Psi>=
\sup_{\Psi\in\Hil: \|\Psi\|\leq1}<\Psi,N_q\Psi>=\sup_{d\in
	l^2:\sum|d_k|^2\leq1}\sum_k|d_k|^2\gamma_{k-1}^2.
$$
Here we have written $\Psi=\sum_kd_k\Phi_k$, using the fact that the vectors
$\Phi_k$'s form an orthonormal basis of $\Hil$. With $l^2$ we  indicate 
the usual Hilbert space of the square-integrable sequences. Therefore we have:
\begin{itemize} \item{}
	$q=1 \Rightarrow \gamma_{k-1}^2=k$. This implies that $\|a_q\|=\infty$.
	\item{}
	$q=-1 \Rightarrow \gamma_{k-1}^2=0,1$ depending on whether the index $n$ is even
	or odd. However, in both cases, $\gamma_{k-1}^2\leq 1$, and therefore
	$$\|a_q\|^2\leq\sup_{d\in
		l^2:\sum|d_k|^2\leq1}\sum_k|d_k|^2=1.
	$$
	Of course for fermions we have in fact $\|a_q\|=1$.
	\item{}
	for general $q\in]-1,1[$ we see that $\gamma_{k-1}^2\leq \frac{2}{1-q}$,
	independently of $n$. As a consequence, it can be shown that $\|a_q\|$ is bounded
	by $\frac{2}{1-q}$.
\end{itemize}
Of course, if $a_q$ is bounded, then $a_q^\dagger$ and $N_q=a_q^\dagger a_q$ are bounded as well.

Many more information on quons can be found in \cite{moh,fiv,gree}. What is relevant for us here is to notice that quons give rise to bounded ladder operators acting on an infinite-dimensional Hilbert space.


\begin{thebibliography}{99}
	
	
	
			\bibitem{hall} B. C. Hall, {\em Quantum theory for mathematicians}, Springer, New York (2013)
	
	
	\bibitem{mess} A. Messiah, {\em Quantum mechanics}, vol. 2, North Holland Publishing Company, Amsterdam, (1962)
	
	
	\bibitem{denittis} G. De Nittis, D. Polo Ojito, {\em About the notion of eigenstates for C*-algebras and some application in quantum mechanics}, J. Math. Phys., {\bf 64}, 083506 (2023)
	
	\bibitem{br} O. Bratteli and D.W. Robinson, {\em Operator
		algebras and Quantum statistical mechanics 1}, Springer-Verlag, New
	York, (1987)
	
	\bibitem{br2} O. Bratteli and D.W. Robinson, {\em Operator
		algebras and Quantum statistical mechanics 2}, Springer-Verlag, New
	York, (1987)
		\bibitem{nb} E. Nelson, {\em Note on non-commutative integration}, J. Funct. Anal., {\bf 15} (1974)
	103-116
	\bibitem{sewbook1} G.L. Sewell, {\it Quantum Theory of Collective Phenomena}, Oxford University Press,
	Oxford (1989)
	
	\bibitem{sewbook2} G.L. Sewell, {\it Quantum Mechanics and its Emergent Macrophysics}, Princeton University Press,
	(2002)
	
	
	
\bibitem{mariacamillo} M. Fragoulopoulou, C. Trapani, {\em Locally convex quasi *-algebras and their representations}, Lecture Notes in Mathematics 2257, Springer Nature Switzerland (2020)
	
	
	\bibitem{aitbook}  J.-P.
	Antoine, A. Inoue and  C. Trapani {\it Partial *-algebras and Their
		Operator Realizations}, Kluwer, Dordrecht, 2002
	
	\bibitem{bag1} F. Bagarello, {\em  Applications of Topological *-Algebras of Unbounded
		Operators}, J. Math. Phys., {\bf 39}, 2730 (1998)
	

	\bibitem{schu} K. Schm\"udgen, {\it Unbounded operator algebras
		and Representation theory}, Birkh{\"a}user, Basel, 1990
	
	
	
	
	
	
	\bibitem{moh} R.N. Mohapatra, {\em Infinite statistics and a possible small
		violation of the Pauli principle}, Phys. Lett. B, {\bf 242}, 407-411, (1990)
	
	\bibitem{fiv} D.I. Fivel, {\em Interpolation between Fermi and Bose
		statistics using generalized commutators}, Phys. Rev. Lett., {\bf 65},
	3361-3364, (1990); Erratum, Phys. Rev. Lett., {\bf 69},
	2020, (1992)
	
	\bibitem{gree} O.W. Greenberg, {\em Particles with small violations of Fermi
		or Bose statistics}, Phys. Rev. D, {\bf 43}, 4111-4120, (1991)
	
	
	\bibitem{fc} F.Bagarello, C.Trapani, {\em CQ*-algebras: structure properties},
	Publ. RIMS, Kyoto Univ., {\bf 32}, 85-116, (1996)
	
		\bibitem{cqellepi} F. Bagarello, C. Trapani and S. Triolo, {\em Quasi *-algebras of measurable operators}, Studia Mathematica,
	{\bf 172}, 289-305 (2006)
	
	
		\bibitem{klauder} J. R. Klauder, B. S. Skagerstam Eds., {\em Coherent states. Applications in physics and mathematical physics}, World Scientific, Singapore (1985)
	
	\bibitem{aagbook}  S.T. Ali, J-P.  Antoine and  J-P.  Gazeau,
	{\em  Coherent States, Wavelets and Their Generalizations\/},
	Springer-Verlag, New York, (2000).
	\bibitem{nelson}E. Nelson, {\it Note on non-commujtative integration}, J. Funct. Anal., {\bf 15}, 103-116 (1974)
	\bibitem{didier} M. Combescure,  R. Didier, {\em Coherent States and Applications in Mathematical Physics},   Springer, (2012)
	
	\bibitem{gazeaubook}  J-P.  Gazeau, {\em Coherent states in quantum physics}, Wiley-VCH, Berlin (2009)
	
	\bibitem{perelomov} A. M. Perelomov, {\em Generalized coherent states and their applications}, Springer-Verlag, Berlin (1986)
	
	\bibitem{BAAG}  S. T. Ali, J.-P. Antoine, F. Bagarello, J.-P. Gazeau, Guest Editors, {\em Coherent states: mathematical
		and physical aspects}, Journal of Physics A: Mathematical and Theoretical, Special Issue, {\bf 45}, N. 24, (2012)
	
	\bibitem{ABG}   J.-P. Antoine, F. Bagarello, J.-P. Gazeau Eds, {\em Coherent states	and applications: a contemporary panorama}, Springer Proceedings in Physics, (2018)
		
	
	\bibitem{bagspringer} F. Bagarello, {\em Pseudo-bosons and their coherent states}, Springer (2022)
	
	
		\bibitem{mosta} A. Mostafazadeh, {\em Pseudo-hermitian quantum mechanics},  Int. J. Geom. Methods Mod. Phys., {\bf 7}, 1191-1306 (2010)
	
	
	
	\bibitem{specissue2012} C. Bender, A. Fring, U. G\"nther, H. Jones Eds, {\em Special issue on quantum physics with non-Hermitian operators}, J. Phys. A: Math. and Ther., {\bf 45} (2012)
	
	
	\bibitem{bagabook} F. Bagarello, J. P. Gazeau, F. H. Szafraniec e M. Znojil Eds., {\em Non-selfadjoint operators in quantum physics: Mathematical aspects}, John Wiley and Sons (2015)
	
	
	\bibitem{bagprocpa} F. Bagarello, R. Passante, C. Trapani, {\em Non-Hermitian Hamiltonians in Quantum Physics;
		Selected Contributions from the 15th International Conference on Non-Hermitian
		Hamiltonians in Quantum Physics}, Palermo, Italy, 18-23 May 2015, Springer (2016)
	
	\bibitem{bender} C. M. Bender, {\em PT Symmetry in quantum and classical physics}, World Scientific, (2019)
	
	\bibitem{pede} G. K. Pedersen, {\em C*-algebras and their automorphism groups}, Academic Press, London-New York (1979)
	
	
	
	
	
	

	
	
	
%
%
%
%
	
	
	
	
	
	
	
	
	
	
	
	



	


\end{thebibliography}
\end{document}